\documentclass[envcountsame]{llncs}
\usepackage{amsmath, amsfonts, amssymb}
\usepackage{xspace}
\usepackage{upgreek}
\usepackage{graphicx}
\usepackage[stable]{footmisc}
\usepackage{braket}

\renewcommand{\geq}{\geqslant}
\renewcommand{\leq}{\leqslant}
\renewcommand{\epsilon}{\varepsilon}
\renewcommand{\Delta}{\Updelta}

\newcommand{\calA}{\ensuremath{\mathcal{A}}}
\newcommand{\calB}{\ensuremath{\mathcal{B}}}
\newcommand{\calC}{\ensuremath{\mathcal{C}}}

\newcommand{\calR}{\ensuremath{\mathcal{R}}}

\newcommand{\indexing}{\textsc{Indexing}\xspace}
\newcommand{\equality}{\textsc{Equality}\xspace}

\newcommand{\upto}{\ensuremath{\ldots}}

\newcommand{\pregion}{\mbox{p-region}\xspace}
\newcommand{\pregions}{\mbox{p-regions}\xspace}
\newcommand{\prepresentation}{\mbox{p-representation}\xspace}

\newcommand{\LCE}{\ensuremath{\textup{LCE}}}
\newcommand{\LCEP}{\ensuremath{\text{LCE}_\text{p}}\xspace}

\title{Pattern Matching in Multiple Streams\footnote{This work was partially supported by EPSRC.}}
\author{
    Rapha\"{e}l Clifford\inst{1} \and
    Markus Jalsenius\inst{1} \and
    Ely Porat\inst{2} \and
    Benjamin Sach\inst{3}}
\institute{
    University of Bristol, Department of Computer Science, Bristol, UK \and
    Bar-Ilan University, Department of Computer Science, Ramat-Gan, Israel \and
    University of Warwick, Department of Computer Science, Coventry, UK}

\begin{document}

\maketitle


\begin{abstract}

    We investigate the problem of deterministic pattern matching in multiple
    streams. In this model, one symbol arrives at a time and is associated with one of $s$
    streaming texts. The task at each time step is to report if there is a new match between a
    fixed pattern of length $m$ and a newly updated stream.  As is
    usual in the streaming context, the goal is to use as little space
    as possible while still reporting matches quickly.  We
    give almost matching upper and lower space bounds for three distinct
    pattern matching problems.  For exact matching we
    show that the problem can be solved in constant time per arriving symbol and $O(m + s)$ words of space.  For the
    $k$-mismatch and $k$-difference problems we give $O(k)$ time solutions that require $O(m+ks)$ words of space. In all three cases we also give space lower bounds which show our methods are optimal up to a single logarithmic factor.  Finally we set out a number of open problems related to this new model for pattern matching.
\end{abstract}

\thispagestyle{plain}

\section{Introduction}

We introduce a new set of problems centered on pattern matching in multiple streaming
texts. In this model, one symbol arrives at a time and is added to the tail
of exactly one of $s$ streaming texts.   The task at each time step is to
report if there is a new match between a fixed pattern $P$ of length
$m$ and a newly updated stream.  Our interest is in deterministic
algorithms with guaranteed worst case time complexity.  The goal is to
use as little space as possible while still reporting matches quickly.

The problem of pattern matching in a single stream using limited space had
a major breakthrough in 2009 when it was shown that exact
matching can be performed using only $O(\log{m})$ words of space and
$O(\log{m})$ time~\cite{Porat:09}. This result was subsequently simplified~\cite{EJS:2010}
and then improved~\cite{BG:2011} to run in constant time per new
symbol.   To achieve this small space, these methods are however all
necessarily randomised and allow some small probability of
error. Where neither randomisation nor error is permitted, a
straightforward argument shows us that there is no
hope of using space sublinear in the pattern size.  This follows
directly from the observation that we could use such a matching
algorithm to reproduce the pattern in its entirety and that it
therefore must use at least linear space.

Where there are multiple streams however the situation is not as clear cut
even where no randomisation is allowed.  A naive approach would simply
be to store $s$
copies of the working space, one for each stream. This will typically then require $\Theta(ms)$
space overall. Where the number of streams $s$ is large, the space
usage of such an approach is therefore likely to be prohibitive.

Our contribution is first to show that for three particularly common pattern matching
problems, pattern matching in multiple streams can be solved in
considerably less than $\Theta(ms)$ space without the use of randomisation.
In Section~\ref{sec:exact} we give a constant time algorithm for exact matching that requires
only $O(m+s)$ words of space.
In Section~\ref{sec:lce} we review a recently introduced data structure which allows us to perform longest common extension (LCE) queries between a streaming text and a fixed pattern in constant time and $O(m)$ space.  This data structure is then used in Sections~\ref{sec:k-mismatch}
and~\ref{sec:k-diff}  where we give $O(k)$ time and $O(m + ks)$ space solutions to the $k$-mismatch and
$k$-difference problems.  In Section~\ref{sec:space} we
give almost matching space lower bounds for our problems as well as lower bounds for some other common pattern matching problems. Finally in
Section~\ref{sec:open} we set out a number of open problems that immediately arise for this new model of pattern matching.


\section{Related work}

Randomised space lower bounds for a wide range of pattern matching
problems in a single stream were given in~\cite{CJPS:2011}.
In~\cite{CEPP:2011,CS:2011} it was also shown that a large set of
pattern matching algorithms could be converted from offline to online
with only at worst a multiplicative logarithmic factor overhead in
their time complexity. This therefore provided an effective
deamortisation of almost the entire field of combinatorial pattern
matching and a ready tool for the construction of fast streaming
pattern matching algorithms.

In the more usual offline setting, a great deal of progress has been made in finding fast  algorithms for
a variety of approximate matching problems. One of the most
studied is the Hamming distance which measures the number of
single character mismatches between two strings.  Given a text of length $n$ and a
pattern of length $m$, the task is to report the Hamming distance at
every possible alignment.   Solutions running in $O(n\sqrt{m\log{m}})$ time which are based
on repeated applications of the FFT were given independently by both
Abrahamson and Kosaraju in 1987~\cite{Abrahamson:1987,Kosaraju:1987}.
Particular interest has been paid to the bounded variant we also
consider called the
$k$-mismatch problem.  Here a bound $k$ is given and we need only
report the Hamming distance if it is less than or equal to $k$.  If  the number of mismatches is greater than the bound, the algorithm need only report that fact and not give the actual Hamming distance. In 1986 Landau and Vishkin gave a  beautiful $O(nk)$ time algorithm that is not FFT based and uses constant time
lowest common ancestor (LCA) operations on the suffix tree of the pattern and
text~\cite{LV:1986a}. This was subsequently  improved to $O(n \sqrt{k\log{k}})$ time by a method based on filtering and FFTs again~\cite{ALP:2004}.  A separate line of research considered
the question of how to find approximations within a $(1+\epsilon)$ multiplicative factor of the Hamming distance~\cite{Indyk:1998,Karloff:1993}. The edit distance measures the minimum number of single character insert, delete and mismatch operations required to transform one string into another. We  consider a bounded problem called $k$-difference which reports the edit distance at every alignment only if it is at most $k$. An $O(nk)$ solution to this problem was given in~\cite{LV:1988} using LCE queries to perform jumps within the dynamic programming table.

\subsection{Preliminaries and a new model for multiple streams}

In order to design algorithms for pattern matching in multiple streams
we first define a new computational model which we will use throughout. In this new model there is only one stream, however
we will carefully distinguish between space that any pattern matching
algorithm uses that is associated with the pattern and
space associated with the text.  We call this model, which may be of
independent interest, the \emph{single stream read-only
  preprocessing model} or simply the \emph{read-only preprocessing model} for short.

Any algorithm that operates in this model operates in two phases: a
preprocessing phase followed by an online phase. During the
preprocessing phase, the algorithm processes the pattern without any
knowledge of the text. The output of the preprocessing phase is termed
the \emph{pattern~space}.   The online phase
is provided read-only access to a copy of the pattern space. It is
important to note that the online phase is not provided direct access
to the pattern unless the pattern is explicitly stored in the pattern
space. The online phase continues as in the original single stream model, processing each text character as it arrives. Any space used by the online phase in addition to the read-only pattern space is termed \emph{text~space}. The text space can therefore only store information that is associated with the text and its relation to information already stored in the pattern space.

Algorithms developed in the read-only preprocessing model translate directly to the
multiple stream model that we are interested in. Consider an algorithm
in the read-only preprocessing model. Let $f$ be the
time taken per new arriving symbol, $g_\textup{p}$ the pattern space
and $g_\textup{t}$ the text space of this algorithm. Since the pattern
space is independent of the text, it can be shared across multiple
texts. Therefore we can directly derive a new pattern matching
algorithm in the multiple stream model which runs in $O(f)$ time per
character using $O(g_\textup{p}+sg_\textup{t})$ space.  The space and
time requirements of the preprocessing stage are arguably of less
significance given that they are a function only of the pattern size.
Nevertheless, for all three pattern matching problems we consider,
the preprocessing stage can be implemented in $O(m)$ space and $O(m\log{m})$ time.
We assume throughout that all computation is performed in the unit cost RAM model.

\section{Exact matching}\label{sec:exact}

We begin by giving a real-time variant of the KMP algorithm. We will apply it to the read-only preprocessing model ensuring it takes
$O(1)$ time, using $O(m)$ pattern space and $O(1)$ text space,
so that in the multiple stream model it will take $O(1)$ time and use $O(m+s)$ space.
The more usual real-time variant provided by Galil~\cite{Galil:1981} will not suffice for our purposes as it buffers the text and therefore would use $O(ms)$ space in the multiple stream model.

First recall that at each stage of the standard KMP algorithm we keep track of the longest prefix of the pattern that matches a suffix of the text seen so far. When a new character $T[i]$ arrives, we extend the matching prefix by $T[i]$ if possible, otherwise we shift the pattern according to a pre-computed prefix table, also known as the failure function. More precisely, suppose that $P[0\upto j-1]$ matches $T[i-j\upto i-1]$ when $T[i]$ arrives. If $P[j]=T[i]$, we extend the match, otherwise we look at position $j$ of the prefix table, which gives us the largest value $0<j'<j$ such that $P[0\upto j'-1]$ matches $T[i-j'\upto i-1]$ and $P[j']\neq P[j]$. Then we compare $P[j']$ with $T[i]$ to see if we can extend the new match by $T[i]$. If not, we shift the pattern again using the prefix table, and so on.

While it is well-known that the time complexity per
character is $O(1)$ amortised, our motivations call for an unamortised
solution. Instead of using a prefix table, for each position $j$ of the pattern we store a dictionary $D_j$ that contains every pair $(\sigma, j')\in \Sigma \times \{1,\dots,m-1\}$ where $0<j'<j$ is the largest index such that $P[0\upto j'-1]$ matches $P[j-j'\upto j-1]$ and $P[j']=\sigma \neq P[j]$.
The pairs $(\sigma, j')$ in $D_j$ are indexed by the symbol $\sigma$ and $\Sigma$ denotes the alphabet.
In other words, whenever a match $P[0\upto j-1]$ cannot be extended to $P[0\upto j]$ (i.e.,~$P[j]\neq T[i]$), instead of repeatedly shifting the pattern according to the prefix table until $T[i]$ is matched, we look up symbol $T[i]$ in $D_j$ and immediately get the length $j'$ of the prefix of $P$ that the KMP algorithm would eventually align with $T[i-j'\upto i-1]$ in order to be able to extend the match to $P[j']=T[i]$. The dictionaries $D_j$ can be pre-computed as all shifts are based on self-matches with the pattern itself. By using a static perfect dictionary we can do lookups in constant time. Thus, a KMP algorithm equipped with
these dictionaries instead of a prefix table will run in unamortised $O(1)$ time per character.
An interesting fact is that the total space usage to store these dictionaries is only $O(m)$. The following lemma was proved in~\cite{Simon:1993}, where it was stated in the language of finite automata. For completeness we include a proof.

\begin{lemma}[Theorem 1 of~\cite{Simon:1993}]
    \label{lem:shifts}
    The sum of the sizes of all dictionaries, \[ \sum_{j=0}^{m-1} |D_j| \leq m\,.\]
\end{lemma}
\begin{proof}
    A lookup in the dictionary $D_j$ results in a shift of the pattern. We show that for every length $\ell\in\{1,\dots,m-1\}$ there is at most one element over all dictionaries that moves the pattern along by $\ell$ positions. For contradiction, suppose that some $(\sigma_1,j'_1)\in D_{j_1}$ and $(\sigma_2,j'_2)\in D_{j_2}$, where $j_1<j_2$, both shift the pattern by $\ell$. By definition, we have $P[j_1]\neq P[j'_1]$. If the same shift $\ell$ is applied when the $j_2$-length prefix of $P$ is moved along, we must have $P[j_1]=P[j_1-\ell]=P[j'_1]$, which leads to a contradiction.
    \qed
\end{proof}

By using static perfect hashing, we can store all dictionaries $D_0,\dots,D_{m-1}$ in $O(m)$ pattern space. Although preprocessing time is not our focus, we briefly discuss how to construct the dictionaries in $O(m \log \log m)$ time.

We begin by constructing the standard KMP prefix table in $O(m)$ time. We then construct each dictionary $D_j$ by considering $j$ in increasing order, starting with $j=0$. For any $j\geq 0$, the dictionary $D_j$ is constructed as follows. Let $j'<j$ be the index given by the original KMP prefix function for $P[0 \ldots j]$. The elements in
$\Set{ (\sigma, j'') \in D_{j'} | \sigma \neq P[j] } \cup \{(P[j'], j')\}$ are added to $D_j$. These are precisely the elements that belong to $D_j$ according to its definition. Over all~$j$, gathering the elements to be added to the dictionaries takes $O(m)$ time. The running time is therefore dominated by the time it takes to insert the elements into the dictionaries. By using the the static dictionary of Ru\v{z}i\'{c}~\cite{Ruzic:08}, construction takes
\begin{equation*}
    \sum_{j=0}^{m-1} O\big(|D_j| \log{\log {|D_j|}}\big) \,\leq\, O(m \log\log m) \,.
\end{equation*}
time, where Lemma~\ref{lem:shifts} has been applied.

The next lemma summarises the result, which together with the properties of the read-only preprocessing model gives us Theorem~\ref{thm:exact}.

\begin{lemma}
    Exact matching can be solved in the read-only preprocessing model in $O(1)$ time per character and using $O(m)$ pattern space and $O(1)$ text space.
\end{lemma}

\begin{theorem}
    \label{thm:exact}
    Exact matching in the multiple stream model with $s$ texts can be solved in $O(1)$ time per character and $O(m+s)$ space.
\end{theorem}

\section{LCE queries in a stream}\label{sec:lce}

In preparation for the algorithms we give for the $k$-mismatch and
$k$-difference problems, we will be required to maintain a data
structure that allows us to
compute LCE queries in a streaming text in $O(1)$ time  and $O(m)$ space.  This method was first introduced
in~\cite{CS:2010}, although the idea of representing the text in terms of substrings of the pattern was first used in a different
setting in~\cite{ALLS:2007}. We provide a brief recap here.

We will split the streaming text into contiguous substrings which are encoded as triplets, $(i',j,\ell)$, each representing an $\ell$-length text substring $T[i'\upto (i'+\ell-1)]$ that equals a pattern substring $P[j\upto (j+\ell-1)]$. We refer to such a triple as \emph{\pregion} and a disjoint ordered sequence of triplets that encode the entire text as a \emph{\prepresentation}. For example, with $P=\texttt{babbac}$ and $T=\texttt{abcaababba}$, a \prepresentation of $T$ is \[(0,1,2),(2,5,1),(3,4,1),(4,1,2),(6,1,4).\] The \prepresentation is not necessarily unique. We say that it is of minimal length if it contains a minimal number of triplets.

In~\cite{CS:2010} it was shown that we can extend a minimal length \prepresentation of $T[0\upto i-1]$ to a minimal length \prepresentation of $T[0\upto i]$ in $O(1)$ time when symbol $T[i]$ arrives. More precisely, the extended \prepresentation is obtained greedily by updating the last \pregion if possible, otherwise adding a new \pregion $(i,j,1)$, where $j$ is some position such that $P[j]=T[i]$.
To accomplish this task in $O(1)$ time, a suffix tree of the pattern
can be used~\cite{CS:2010}. For our purposes in terms of pattern space
and text space, we may construct the suffix tree during the pattern
preprocessing phase and store it together with the pattern itself in
the pattern space. Deciding whether to update the last \pregion or add
a new one when a new symbol $T[i]$ arrives can then be done in
constant text space. For simplicity of explanation we will assume that
all symbols in $T$ occur at least once~in~$P$.  For further details of
the method and how to handle symbols which occur in the text but not
the pattern, we refer the interested reader to~\cite{CS:2010}.

As we will see shortly, a benefit of using a minimal length \prepresentation of $T$ is that we can answer longest common extension (LCE) queries between the pattern and $T[0\upto i]$ in $O(1)$ time. We write $\LCE(i',j)$ to denote the length of the longest prefix of $P[j\upto m-1]$ that is also a prefix of $T[i'\upto i]$.
The following lemma was stated as Lemma~1 in~\cite{CS:2010}.

\begin{lemma}
    \label{lem:three-regions}
    For a minimal length \prepresentation of $T$, at most three \pregions overlap $T[i'\upto (i'+\ell-1)]$, where $\ell$ is the length returned by any $\LCE(i',j)$ query.
\end{lemma}

It is well known that we can precompute a static data structure using
$O(m)$ space, denoted \LCEP, to support LCE queries between the
pattern and itself in $O(1)$ time. This is traditionally achieved with
a suffix tree on which lowest common ancestor queries are answered in
constant time. It now follows from Lemma~\ref{lem:three-regions} that
any $\LCE(i',j)$ query on the streaming text can be answered in $O(1)$ time by performing at most three pattern-pattern LCE queries in the \LCEP structure.

\section{$k$-mismatch in multiple streams}\label{sec:k-mismatch}

We described in Section~\ref{sec:lce} how a \prepresentation of the
text can be maintained in $O(1)$ time and $O(m)$ pattern space. The
actual \prepresentation of $T[0\upto i]$ will be stored in the text
space of size $O(i)$. Instead of storing every \pregion, we could
store only the most recent \pregions of the text. This representation
will of course only give us access to some suffix of the text seen so
far. Our algorithm for the $k$-mismatch problem, which we present in the read-only preprocessing
model in the first instance, will store the most
recent $4(k+1)$ \pregions, requiring $O(k)$ text space.

In order to determine whether $T[(i-m+1)\upto i]$ has at most $k$ mismatches with $P$ when $T[i]$ arrives, we apply the kangaroo technique~\cite{LV:1985} consisting of up to $k+1$ LCE queries between the text and the pattern. We will perform these LCE queries in the \emph{reverse} of $T$ and $P$, starting from the rightmost character $T[i]$. It should not be too hard to see that the data structures described in Section~\ref{sec:lce} can be modified to support reverse LCE queries between the pattern and the text with no effect on the asymptotic time and space complexities. Also, Lemma~\ref{lem:three-regions} holds for reverse LCE queries. In the next lemma we prove that all LCE queries performed by the algorithm fall within the $4(k+1)$ most recent \pregions.

\begin{lemma}
    \label{lem:k-mismatch}
    The $k$-mismatch problem can be solved in the read-only
    preprocessing model in  $O(k)$ time per character and using $O(m)$ pattern space and $O(k)$ text space.
\end{lemma}
\begin{proof}
    The time complexity follows immediately from the algorithm description, similarly with the pattern and text space complexities. For the correctness of the algorithm we need to show that none of the at most $k+1$ reverse LCE queries performed to determine the number of mismatches between $T[(i-m+1)\upto i]$ and $P$ fall outside the text substring represented by the $4(k+1)$ most recent \pregions.

    From Lemma~\ref{lem:three-regions} it follows that one LCE query could span three \pregions. As part of the kangaroo technique, we skip over a mismatch position between each LCE query. The mismatch could fall inside a new \pregion, hence no more than a total of $4(k+1)$ \pregions are involved in a series of up to $k+1$ LCE queries.
    \qed
\end{proof}

From Lemma~\ref{lem:k-mismatch} and the properties of the read-only
preprocessing model, we immediately have the following theorem.

\begin{theorem}
    The $k$-mismatch problem can be solved in the multiple stream model with $s$ texts in $O(k)$ time per character and $O(m+ks)$ space.
\end{theorem}

\section{$k$-difference in multiple streams}\label{sec:k-diff}

Let $D[j,i]$ denote the minimum of all $k$-bounded edit distances between the pattern prefix $P[0\upto j]$ and all suffixes of $T[0\upto i]$. For the $k$-difference problem we want to output $D[m-1,i]$ as soon as $T[i]$ arrives. We have the standard dynamic programming recurrence,
\begin{equation*}
D[j,i] = \text{min} \left \lbrace
    \begin{array}{ll}
        D[j,\, i-1]+1 & \text{ (insert)}\\
        D[j-1,\, i]+1 & \text{ (delete)} \\
        D[j-1,\,i-1] + 1 - \text{eq}(i,j) & \text{ (mismatch)} \\
        k+1 & \text{ ($k$-bounded)}\\
    \end{array} \right.
\end{equation*}
where $\text{eq}(i,j)=1$ if $T[i]=P[j]$ and $0$ otherwise. For the base cases we have $D[j,-1]=\min(k+1,j+1)$ and $D[-1,i]=0$ for all $i,j$.

We now present a solution for the $k$-difference problem, first in
the read-only preprocessing model and then give the final result in
the multiple stream model as required.  Whenever a text character $T[i]$ arrives such that $i$ is a multiple of $k$, we start a \emph{child process} which will be responsible for outputting $D[m-1,i']$ for all $i'\in \{(i+k),\ldots,(i+2k-1)\}$ as each such $T[i']$ arrives (Interval~2 in Fig.~\ref{fig:dynamic}). Therefore, there is a child process responsible for each and every output. The $k$ text arrivals between $T[i]$ and the first output for a child process will, as explained below, give us enough time to prepare for the outputs. Observe that at most two child processes are running at any one time, hence we only need to show that one child process can be implemented in $O(k)$ time per character,  $O(m)$ pattern space and $O(k)$ text space.
\begin{figure}[ht]
    \hskip 45pt \includegraphics[scale=1.0]{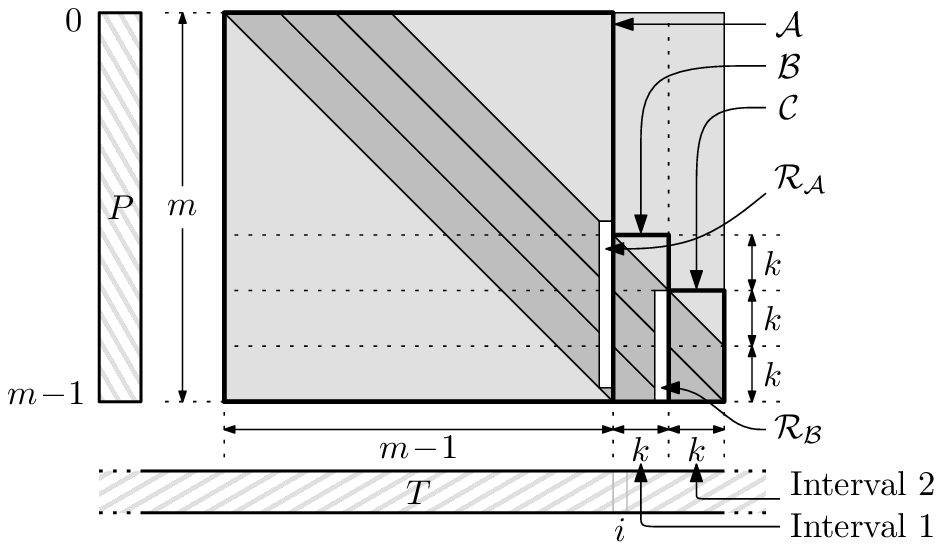}
    \caption{\label{fig:dynamic} The dynamic programming table for $k$-difference.}
\end{figure}

The operation of a child process is divided into three stages, each responsible for computing cells of the blocks denoted $\calA$, $\calB$ and $\calC$, respectively, of the dynamic programming table in~Fig.~\ref{fig:dynamic}. In the first stage, which runs during the $k/2$-length text interval starting with the arrival of $T[i]$ (first half of Interval~1 in Fig.~\ref{fig:dynamic}), the child process will compute the cells marked $\calR_\calA$ in Fig.~\ref{fig:dynamic}, that is the cells $D[(m-3k-1)\upto (m-2),\,i-1]$. To do this we use the technique of Landau-Vishkin~\cite{LV:1988} which is an offline algorithm that we run on block $\calA$. Their algorithm takes $O(k^2)$ time and uses $O(m)$ space by operating along the diagonals of the dynamic programming table (shaded in dark gray in the figure). Their algorithm is based around LCE queries of the text, which we can take advantage of in a similar fashion to the $k$-mismatch algorithm from the previous section. We will see below that we only need to store $O(k)$ \pregions of the text to perform the LCE queries, hence we can run their algorithm in $O(k)$ text space. The running time of $O(k^2)$ is spread evenly over $k/2$ text arrivals, and therefore takes $O(k)$ time per character.

Once the first stage is completed we have obtained the values of the cells marked $\calR_\calA$ in the figure. Now the second stage takes over. Here we compute the cells of block $\calB$ of the dynamic programming table by direct use of the recurrence above. The work is spread evenly over the next $k/2$ text arrivals (second half of Interval~1) by computing two columns of the block $\calB$ per arrival. The final column, marked $\calR_\calB$ in the figure, is therefore done by the arrival of $T[i+k]$. Thus, the second stage takes $O(k)$ time per character and uses only $O(k)$ space. We should mention that the values of cells on the boundary of $\calB$ (excluding $\calR_\calA$) that are used in the recurrence when computing block $\calB$ are all set to $\infty$. This however does not affect the correctness of the computed values of $\calR_\calB$.

In the final and third stage we compute all cell values of block~$\calC$ by direct use of the recurrence, like we did for block $\calB$ during the second stage. We use the values of $\calR_\calB$ and set the other boundary cells to $\infty$ (which does not affect the correctness of the final output). For each arriving character during Interval~2 in Fig.~\ref{fig:dynamic}, we compute the corresponding column of  block $\calC$. Therefore the running time is $O(k)$ per character, and space usage is $O(k)$.

All steps of the child process described above, except for the real-time LCE processing, require only $O(k)$ space. During the pattern preprocessing phase, as for the $k$-mismatch algorithm above, we will construct the required $\LCEP$ structure and the suffix tree of the pattern, and store both these structures as well as a copy of the pattern, using a total of $O(m)$ pattern space.
We modify the LCE processing to store only the most recent $5(k+1)$ \pregions.
We will show that whenever the edit distance is $k$ or less, storing the $5(k+1)$ most recent \pregions is sufficient to support every LCE query. Thus, if any LCE query stretches beyond these \pregions, the edit distance must be more than $k$. The following lemma summarises the result. The result for the multiple stream model follows immediately from the above observation about the relationship between the models.


\begin{lemma}
    \label{lem:k-diff}
    The $k$-difference problem can be solved in the read-only
    preprocessing model in  $O(k)$ time per character using $O(m)$ pattern space and $O(k)$ text space.
\end{lemma}

\begin{proof}[sketch]
    The time and space complexities follow from inspection of the algorithm description.
    For correctness, suppose there exists an $i'\in\{(i+k), \dots, (i+2k-1)\}$ such that $D[m-1,i'] \leq k$. Therefore, by the problem definition, there exists an $\ell$ such that $P$ can be transformed into $T[(i'-\ell+1)\upto i']$ in at most $k$ insert, delete and mismatch operations. We will first show that this immediately implies the existence of a \prepresentation of $T[(i'-\ell+1)\upto i']$ containing at most $2k+1$ \pregions.

    Consider any transformation of $P$ into $P':=T[(i'-\ell+1)\upto i']$ which contains at most $k$ operations. We denote by $C$ the $\ell$-length array which states the `origin' of each character in $P'$: $C[j']=j$ if the transformation aligns $P'[j']$ with $P[j]$ and $P'[j'] = P[j]$, otherwise $C[j']=-\infty$, which means that $C[j']$ is the result of an insert operation or is aligned with a symbol different from $P'[j']$.

    We can construct a \prepresentation $R$ of $P'$ by a single pass of $C$ as follows.
    If $C[0]\neq -\infty$, we begin by creating the \pregion $(0,C[0],1)$. Otherwise, we create the region $(0,j'',1)$ where $j''$ is any index such that $P[j'']=P'[0]$.
    For each $j'>0$, we consider three disjoint cases:

    \begin{enumerate}
        \item $C[j'] = C[j'-1]+1$. Increase the length of the most recent region by one.
        \item $C[j'] > C[j'-1]+1$. Create a new \pregion $(j',C[j'],1)$.
        \item $C[j']=-\infty$. Create a new \pregion $(j',j'',1)$, where $j''$ is any index such that $P'[j']=P[j'']$.
    \end{enumerate}




    We now consider the number of \pregions in the \prepresentation $R$. An additional \pregion is only created when either case~2 or case~3 occurs in the construction. Case~3 occurs only when $C[j']=-\infty$. However, by the definition of $C$, each $-\infty$ corresponds $P[j']$ to being a result of a mismatch or insert operation of which there are at most $k$ in total.
    Therefore case~3 occurs at most $k$ times. Case~2 occurs when $C[j'] > C[j'-1]+1$. By the definition of $C$, this implies that either some character $P[j'']$ with $C[j'-1] < j'' < C[j']$ was deleted or $C[j'-1]= - \infty$ which in turn implies that either a mismatch or insert operation occurred at $C[j'-1]$. Hence case~2 can occur at most $k$ times. The total number of \pregions is therefore upper bounded by $2k+1$ as a mismatch or insert can cause two new p-regions to be created, one at some $C[j'-1]$ and another at $C[j']$.

    To see that $5(k+1)$ \pregions are enough, first observe that every LCE query performed ends in the substring $T[(i-m+1) \ldots (i+2k-1)]$.
    As $P$ can be transformed into $T[(i'-\ell+1)\upto i']$ in at most $k$ moves, we have that $\ell\geq m-k$. We also have that $i'\in\{(i+k), \dots, (i+2k-1)\}$ and hence $T[(i'-\ell+1)\upto i']$ is a substring of $T[(i-m+1) \ldots (i+2k-1)]$. We can then then convert the \prepresentation  of  $T[(i'-\ell+1)\upto i']$ into a \prepresentation of $T[(i-m+1) \ldots (i+2k-1)]$ by adding at most $3k$ \pregions of length one each. Thus, $(2k+1)+3k\leq 5(k+1)$ \pregions suffice to support any LCE query.    \qed
\end{proof}


\begin{theorem}
    The $k$-difference problem can be solved in the multiple stream model with $s$ texts in $O(k)$ time per character and $O(m+ks)$ space.
\end{theorem}

\section{Space lower bounds}\label{sec:space}

In this section we show that our space upper bounds for $k$-mismatch
and $k$-difference are optimal (up to a log factor). We also show
that for several other common distance measures, pattern matching in $s$ streams requires $\Omega(ms)$ bits of space, implying that we may not do better than treating each stream independently.

The log sized gap between our lower bounds and upper bounds comes from
the fact that we state the lower bounds in bits whereas the upper
bounds are given in words. A smaller gap could be obtained by
considering large alphabets (see e.g.\@~\cite{CJPS:2011}), however for
simplicity we give our lower bounds assuming binary alphabets.

Our results are based on reductions from two one-way communication complexity problems with known lower bounds.
In a one-way communication model, only Alice can send messages to Bob and then Bob must output the correct answer.
In the \equality problem, Alice has a string $X\in\{0,1\}^{m}$ and Bob has a string $Y\in\left\{ 0,1\right\} ^{m}$. Bob must determine whether $X=Y$. The communication complexity is $\Omega(m)$ bits~\cite{Kushilevitz:97}.
In the \indexing problem, Alice has a string $X\in\{0,1\}^{m}$ and Bob has an index $n\in\{0,\ldots, m-1\}$. Bob must find $X[n]$. The problem is known to have an $\Omega(m)$ bit lower bound~\cite{Kushilevitz:97}.

\begin{theorem}
    \label{thm:k-lower}
    The $k$-mismatch and $k$-difference problems in $s$ streams both require $\Omega(m+ks)$ bits of space.
\end{theorem}
\begin{proof}
    First consider the case where $m\geq ks$. We reduce from the \equality problem, where Alice has a string $X\in\left\{ 0,1\right\}^{m}$ and Bob has a bit string $Y\in\left\{ 0,1\right\} ^{m}$.
    Let the pattern $P$ be the string $X$. Let $A$ be any algorithm that solves either $k$-mismatch or $k$-difference on the pattern $P$. Alice sends the internal state of $A$ to Bob, who feeds the algorithm with the string $Y$ in one of the streams. The output is~0 if and only if $X=Y$, hence $\Omega(m)$ bits of space is required.

    Now consider the case where $m<ks$. We reduce from the \indexing problem, where Alice has a string $X\in\left\{ 0,1\right\}^{ks}$ and Bob has an index $n\in\{0,\ldots, ks-1\}$. Let $A$ be any algorithm that solves either $k$-mismatch or $k$-difference on the pattern $P=\{0\}^m$. Alice feeds each of the $s$ streams with $k$ bits from her string $X$ such that the first stream is fed the first $k$ bits of $X$, the second stream is fed the next $k$ bits of $X$, and so on. Alice then sends the internal state of $A$ to Bob who now wants to determine $X[n]$ which was fed into stream $r=\lfloor n/k\rfloor$. Bob feeds the stream $r$ with $m-k+(n\bmod k)$ 0s, which ensures that $X[n]$ is aligned with the first position of $P$. Let $d$ be the output by $A$ at this alignment and observe that for both $k$-mismatch and $k$-difference, $d$ equals the number of 1s in the last $m$ symbols of the stream $r$. Bob now feeds another 0 into the stream $r$. Let $d'$ be the new output by $A$. It follows that $X[n]=0$ if and only if $d=d'$, hence $\Omega(ks)$ bits of space is required.    The space lower bound of $\Omega(m+ks)$ bits is obtained by combining the two cases.
    \qed
\end{proof}

\begin{theorem}
    Exact matching in $s$ streams requires $\Omega(m+s)$ bits of space.
\end{theorem}
\begin{proof}
    Following the proof of Theorem~\ref{thm:k-lower}, we reduce from the \equality problem if $m>s$, otherwise we reduce from \indexing where Alice feeds each stream with either one 0 or one 1. By feeding $m-1$ 0s into any stream, Bob can determine any of the $s$ bits.
    \qed
\end{proof}

We now turn our attention to the $L_{1}$, $L_{2}$ and Hamming distance problems. For any constant $p$, the $L_p$ distance between two equal length strings $X$ and $Y$ is given by $d_p(X,Y)=\big(\sum_{j} |X[j]-Y[j]|^p\big)^{1/p}$.

\begin{theorem}
    Computing the $L_{1}$, $L_{2}$ and Hamming distances, as well as the cross-correlation/convolution in $s$ streams, requires $\Omega(ms)$ bits of space.
\end{theorem}
\begin{proof}
    We reduce from the \indexing problem, where Alice has a string $X\in\left\{ 0,1\right\}^{ms}$ and Bob has an index $n\in\{0,\ldots, ms-1\}$. Let $A$ be any algorithm that solves either of the problems in the statement of the lemma on the pattern $P=\{1\}^m$, where instead of computing the
    $L_2$ distance, the square of the distance is computed. Observe that a lower bound for computing the square of the $L_2$ distance is also a lower bound for the $L_2$ distance. Each of the problems is now \emph{local} in the sense that the output is the sum of $m$ position-wise values. Following the idea in the proof of Theorem~\ref{thm:k-lower}, Alice feeds each stream with $m$ bits of her string $X$ before sending the internal state to Bob. In order to determine $X[n]$, Bob feeds the appropriate stream $r$ with enough 1s to align $X[n]$ with the first position of $P$. By feeding another 1 into the stream $r$ and comparing the two outputs, Bob can determine the value of $X[n]$.
    \qed
\end{proof}

\section{Open problems}\label{sec:open}

The space complexity of the results we give are tight to within a log factor when no randomisation is permitted. Further, the time complexity of both the exact matching and $k$-difference algorithms we give match that of the fastest known offline algorithms per arriving symbol.  However, for $k$-mismatch there remains a gap of approximately $O(\sqrt{k})$ between the fastest single stream algorithm~\cite{CS:2010} and the time complexity we give for multiple streams. There is an even more pronounced gap for the special case of constant sized alphabets when the bound $k$ is relatively large. Here the $k$-mismatch problem can be solved in a single stream in $O(\log^2 m)$ time per symbol and $O(m)$ space~\cite{CEPP:2008}, independent of the value of $k$. It would be interesting to consider whether there is a $O(\text{poly}(\log m))$ time, multiple stream algorithm using only $O(m+ks)$ space in this case.

We have seen that the read-only preprocessing model is a useful
conceptual tool for developing algorithms in the multiple stream
model. In particular we have used the fact that any efficient algorithm in the former model immediately gives an efficient algorithm in the latter. It is natural to wonder whether these models are in fact equivalent. We conjecture that for any $O(g_\textup{p} + sg_\textup{t})$ space algorithm for a pattern matching problem in the multiple stream model (where $g_\textup{p},g_\textup{t}$ do not depend on $s$), there is an $O(g_\textup{p})$ pattern space, $O(g_\textup{t})$ text space algorithm in the read-only preprocessing model with the same time complexity per character.

If we are concerned with
randomised computation where each output has to be correct with some
(arbitrarily large) constant probability, we can derive new space
lower bounds for all three problems with some modification. In particular, $k$-mismatch and
$k$-difference will require at least $\Omega(\log m + ks)$ bits of
space and exact matching  $\Omega(\log m + s)$ bits of space. These
lower bounds follow from the randomised counterpart of \equality and
\indexing. The randomised one-way communication complexity with private
randomness for \equality is $\Theta(\log m)$ bits~\cite{Yao:79}, and
for \indexing it remains $\Omega(m)$ bits (see~\cite{JKS:08} for an
elementary proof).  It is not yet clear whether these randomised lower bounds for
the multiple streams problem can be met by matching algorithmic upper bounds.

\bibliographystyle{plain}
\bibliography{multbib}

\begin{thebibliography}{10}

\bibitem{Abrahamson:1987}
K.~Abrahamson.
\newblock Generalized string matching.
\newblock {\em SIAM Journal on Computing}, 16(6):1039--1051, 1987.

\bibitem{ALLS:2007}
A.~Amir, G.~M. Landau, M.~Lewenstein, and D.~Sokol.
\newblock Dynamic text and static pattern matching.
\newblock {\em ACM Transactions on Algorithms (TALG)}, 3(2), 2007.

\bibitem{ALP:2004}
A.~Amir, M.~Lewenstein, and E.~Porat.
\newblock Faster algorithms for string matching with \textit{k} mismatches.
\newblock {\em Journal of Algorithms.}, 50(2):257--275, 2004.

\bibitem{BG:2011}
D.~Breslauer and Z.~Galil.
\newblock Real-time streaming string-matching.
\newblock In {\em CPM '11: Proc. 22\textsuperscript{nd} Annual Symp. on
  Combinatorial Pattern Matching}, pages 162--172, 2011.

\bibitem{CEPP:2008}
R.~Clifford, K.~Efremenko, B.~Porat, and E.~Porat.
\newblock A black box for online approximate pattern matching.
\newblock In {\em CPM '08: Proc. 19\textsuperscript{th} Annual Symp. on
  Combinatorial Pattern Matching}, pages 143--151, 2008.

\bibitem{CEPP:2011}
R.~Clifford, K.~Efremenko, B.~Porat, and E.~Porat.
\newblock A black box for online approximate pattern matching.
\newblock {\em Information and Computation}, 209(4):731--736, 2011.

\bibitem{CJPS:2011}
R.~Clifford, M.~Jalsenius, E.~Porat, and B.~Sach.
\newblock Space lower bounds for online pattern matching.
\newblock In {\em CPM '11: Proc. 22\textsuperscript{nd} Annual Symp. on
  Combinatorial Pattern Matching}, pages 184--196, 2011.

\bibitem{CS:2010}
R.~Clifford and B.~Sach.
\newblock Pseudo-realtime pattern matching: Closing the gap.
\newblock In {\em CPM '10: Proc. 21\textsuperscript{st} Annual Symp. on
  Combinatorial Pattern Matching}, pages 101--111, 2010.

\bibitem{CS:2011}
R.~Clifford and B.~Sach.
\newblock Pattern matching in pseudo real-time.
\newblock {\em Journal of Discrete Algorithms}, 9(1):67--81, 2011.

\bibitem{EJS:2010}
F.~Ergun, H.~Jowhari, and M.~Sa{\u{g}}lam.
\newblock Periodicity in streams.
\newblock In {\em RANDOM '10: Proc. 14\textsuperscript{th} Intl. Workshop on
  Randomization and Computation}, pages 545--559, 2010.

\bibitem{Galil:1981}
Z.~Galil.
\newblock String matching in real time.
\newblock {\em Journal of the ACM}, 28(1):134--149, 1981.

\bibitem{Indyk:1998}
P.~Indyk.
\newblock Faster algorithms for string matching problems: Matching the
  convolution bound.
\newblock In {\em FOCS '98: Proc. 39\textsuperscript{th} Annual Symp.
  Foundations of Computer Science}, pages 166--173, 1998.

\bibitem{JKS:08}
T.~S. Jayram, R.~Kumar, and D.~Sivakumar.
\newblock The one-way communication complexity of hamming distance.
\newblock {\em Theory of Computing}, 4(1):129--135, 2008.

\bibitem{Karloff:1993}
H.~Karloff.
\newblock Fast algorithms for approximately counting mismatches.
\newblock {\em Information Processing Letters}, 48(2):53--60, 1993.

\bibitem{Kosaraju:1987}
S.~R. Kosaraju.
\newblock Efficient string matching.
\newblock Manuscript, 1987.

\bibitem{Kushilevitz:97}
E.~Kushilevitz and N.~Nisan.
\newblock {\em Communication complexity}.
\newblock Cambridge University Press, 1997.

\bibitem{LV:1985}
G.~M. Landau and U.~Vishkin.
\newblock Efficient string matching in the presence of errors.
\newblock In {\em FOCS '85: Proc. 26\textsuperscript{th} Annual Symp.
  Foundations of Computer Science}, pages 126--136, 1985.

\bibitem{LV:1986a}
G.~M. Landau and U.~Vishkin.
\newblock Efficient string matching with {$k$} mismatches.
\newblock {\em Theoretical Computer Science}, 43:239--249, 1986.

\bibitem{LV:1988}
G.~M. Landau and U.~Vishkin.
\newblock Fast string matching with $k$ differences.
\newblock {\em Journal of Computer System Sciences}, 37(1):63--78, 1988.

\bibitem{Porat:09}
B.~Porat and E.~Porat.
\newblock Exact and approximate pattern matching in the streaming model.
\newblock In {\em FOCS '09: Proc. 50\textsuperscript{th} Annual Symp.
  Foundations of Computer Science}, pages 315--323, 2009.

\bibitem{Ruzic:08}
M.~Ru{\v{z}}i{\'{c}}.
\newblock Constructing efficient dictionaries in close to sorting time.
\newblock In {\em ICALP '08: Proc. 25\textsuperscript{th} International
  Colloquium on Automata, Languages and Programming}, pages 84--95, 2008.

\bibitem{Simon:1993}
I.~Simon.
\newblock String matching algorithms and automata.
\newblock In {\em First American Workshop on String Processing}, pages
  151--157, 1993.

\bibitem{Yao:79}
A.~C.-C. Yao.
\newblock Some complexity questions related to distributive computing.
\newblock In {\em STOC '79: Proc. 11\textsuperscript{th} Annual ACM Symp.
  Theory of Computing}, pages 209--213, 1979.

\end{thebibliography}


\end{document}